\documentclass[copyright,creativecommons]{eptcs}
\providecommand{\event}{SOS 2021} % Name of the event you are submitting to
\usepackage{underscore}           % Only needed if you use pdflatex.
\usepackage{amsthm}
\usepackage{amsmath}
\usepackage{amssymb}
\usepackage{xspace}
\usepackage{stmaryrd}
\usepackage{enumitem}
\usepackage{bm}
\usepackage{tikz}
\usetikzlibrary{positioning,decorations.pathreplacing,quotes}
\theoremstyle{plain}
\newtheorem{theorem}{Theorem}
\newtheorem{lemma}[theorem]{Lemma}

\theoremstyle{definition}
\newtheorem{definition}[theorem]{Definition}

\def\newarrow#1{\mathop{{\hbox{\setbox0=\hbox{$\scriptstyle{#1\quad}$}{$%
\mathrel{\mathop{\setbox1=\hbox to
\wd0{\rightarrowfill}\ht1=3pt\dp1=-2pt\box1}\limits^{#1}}%
$}}}}}
\newcommand{\Transition}[3]{\ensuremath{#1 \newarrow{#2} #3}}
\newcommand{\States}{\ensuremath{S}}
\newcommand{\Transsys}{\ensuremath{\mathcal{T}}}
\newcommand{\Act}{\ensuremath{\mathcal{A}}}
\newcommand{\Label}{\ensuremath{\ell}}

\newcommand{\Prop}{\ensuremath{\mathcal{P}}}
\newcommand{\Progs}{\ensuremath{\Sigma}}
\newcommand{\Var}{\ensuremath{\mathcal{V}}}

\newcommand{\mutrue}{\ensuremath{\mathtt{t\!t}}}
\newcommand{\mufalse}{\ensuremath{\mathtt{f\!f}}}
\newcommand{\mudiam}[1]{\langle #1 \rangle}
\newcommand{\mubox}[1]{\ensuremath{[ #1 ]}}
\newcommand{\term}{\ensuremath{\tau}}
\newcommand{\mudiabox}[1]{\ensuremath{[\!\!\langle #1 \rangle\!\!]}}

\newcommand{\md}[1]{\ensuremath{\mathit{md}(#1)}}

\newcommand{\sem}[3]{\ensuremath{\llbracket #1 \rrbracket_{#3}^{#2}}}
\newcommand{\Sem}[3]{\ensuremath{|\!| #1 |\!|_{#3}^{#2}}}

\newcommand{\REG}{\textup{REG}}
\newcommand{\CFL}{\textup{CFL}}
\newcommand{\VPL}{\textup{VPL}}
\newcommand{\CSL}{\textup{CSL}}

\newcommand{\PDL}[1]{\def\temp{#1}\textup{PDL}\ifx\temp\empty\else\textup{[}#1\textup{]}\fi\xspace}
\newcommand{\PDLREG}{\PDL{\REG}}
\newcommand{\PDLCFL}{\PDL{\CFL}}
\newcommand{\PDLVPL}{\PDL{\VPL}}

\newcommand{\mucalc}{\ensuremath{\mathcal{L}_\mu}\xspace}

\newcommand{\LTL}{\textup{LTL}\xspace}
\newcommand{\CTL}{\textup{CTL}\xspace}
\newcommand{\RecCTL}{\textup{RecCTL}\xspace}
\newcommand{\FLC}{\textup{FLC}\xspace}
\newcommand{\vpFLC}{\textup{vpFLC}\xspace}
\newcommand{\HFL}[1]{\textup{HFL}\ensuremath{^{#1}}\xspace}

\newcommand{\EXPTIME}{\textup{EXPTIME}\xspace}

\newcommand{\TwoEXPTIME}{\textup{2EXPTIME}\xspace}

\newcommand{\eword}{\varepsilon}
\newcommand{\call}{c}
\newcommand{\ret}{r}
\newcommand{\internal}{i}
\newcommand{\stdown}{\bot}
\newcommand{\aut}[1]{\ensuremath{\mathcal{#1}}}
\newcommand{\lclass}{C}
\newcommand{\deriv}[2]{\ensuremath{\Delta_{#1}(#2)}}

\newcommand{\Nat}{\ensuremath{\mathbb{N}}}

\newlist{proplist}{description}{1}
\setlist[proplist]{leftmargin=1cm, itemindent=-1cm, itemsep=1ex}

\title{Separating the Expressive Power of Propositional Dynamic and Modal Fixpoint Logics}

\author{Eric Alsmann \qquad \qquad Florian Bruse \qquad \qquad Martin Lange
\institute{School of Electrical Engineering and Computer Science\\
University of Kassel, Germany}
\email{ eric.alsmann@student.uni-kassel.de \enskip florian.bruse@uni-kassel.de \enskip martin.lange@uni-kassel.de  }
}
\def\titlerunning{Separating the Expressive Power of PDL and Modal Fixpoint Logics}
\def\authorrunning{E. Alsmann, F. Bruse \& M. Lange}
\begin{document}
\maketitle

\begin{abstract}
We investigate the expressive power of the two main kinds of program logics for complex, non-regular program properties
found in the literature: those extending propositional dynamic logic (\PDL{}), and those extending the modal $\mu$-calculus. This is
inspired by the recent discovery of a decidable program logic called \emph{Visibly Pushdown Fixpoint Logic with Chop} which extends both
the modal $\mu$-calculus and \PDL{} over visibly pushdown languages, which, so far, constituted the ends of two pillars of decidable program logics.

Here we show that this logic is not only more expressive than either of its two fragments, but in fact even more expressive
than their union. Hence, the decidability border amongst program logics has been properly pushed up. We complete the picture
by providing results separating all the \PDL{}-based and modal fixpoint logics with regular, visibly pushdown and arbitrary
context-free constructions. 
\end{abstract}

\section{Introduction}
% !TEX root =  main.tex

\paragraph*{Program Logics.}
Modal logics play a major role in formal program specification and verification, least because modal logics are tightly linked
to the notion of bisimulation-invariance which is deemed to be \emph{the} notion of behavioural equivalence for state-based
programs, resp.\ systems dynamically evolving over time. 

Basic modal logic is inadequate for formal program specification since it lacks the ability to even express the simplest forms
of program correctness like safety (``\emph{no bad state can be reached}''), termination (``\emph{every run finally ends}''), etc.
This is of course due to the fact that a basic modal formula of modal depth $k$ can only ``see'' the next $k$ levels of successors
in a transition system. 

This shortcoming has led to the design and study of several extensions of basic modal logic for formal specification and verification.
Ultimately, the key to retaining bisimulation-invariance but increasing the expressive power to include typical (un-)desired 
program properties is the addition of fixpoint constructs. These come in one of two forms: either as explicit fixpoint quantifiers
-- the most prominent example of such an extension being the modal $\mu$-calculus \mucalc \cite{Kozen83} -- or implicitly in the
form of temporal operators -- the most prominent examples here being temporal logics like \LTL \cite{Pnueli:1977} and \CTL
\cite{ClEm81}. 

Another form of implicit fixpoint operator(s) in an extension of basic modal logic is found in Propositional Dynamic Logic (\PDL{}). 
This comprises, in fact, a family of formalisms, parametrised by generalisations of the accessibility relation in a labelled
transition system (LTS), represented as a class of formal languages. Not surprisingly, the most prominent example of this family is
\PDLREG, typically just called \PDL{}, which can be seen as basic modal logic over an Kleene algebra of accessibility relations in an 
LTS \cite{Fischer79}.

The expressiveness of the logics mentioned so far is well-understood. Most notably, they are all incomparable in expressiveness, 
and all of them can be embedded into \mucalc which is therefore strictly more expressive than any of them. Example properties witnessing
the strictness are also well-known, and their inexpressibility in one of these logics is typically not difficult to prove formally:
\begin{itemize}
\item \PDL{} can only combine eventuality properties with existential path quantification; hence, it cannot express the \CTL-property 
      $\mathsf{AF}q$ stating ``\emph{$q$ holds on all paths at some point}'' (which is also expressible in \LTL);
\item \LTL can only quantify over all paths on the top-level, hence it cannot state $\mathsf{EX}q \wedge \mathsf{EX}\neg q$ stating
      ``\emph{there is a successor satisfying $q$ and one that does not satisfy $q$}'' (which is also expressible in \PDL{} as
      $\mudiam{-}q \wedge \mudiam{-}\neg q$;
\item \CTL cannot state ``\emph{$q$ holds only finitely often on all paths}'' which is possible in \LTL, and it cannot say 
      ``\emph{$q$ holds after every even number of steps}'' which is expressed by the \PDL{} formula $\mubox{(\Sigma\Sigma)^*}q$ 
      over LTS with edge labels from $\Sigma$. 
\end{itemize}

\paragraph*{Non-Regular Program Logics.}
The fact that \mucalc embeds them all means that their expressiveness is limited by \emph{regularity} in the sense that each property
definable in these logics can also be specified by a finite tree automaton or a (bisimulation-invariant) formula of Monadic Second-Order
Logic. While regular expressiveness is sufficient for many program specification and verification tasks in the form of safety, liveness,
fairness properties, there are situations in formal verification where expressiveness beyond regularity is required. This has led to 
the design of specification logics beyond \mucalc or \PDL{}. 
\begin{itemize}
\item A non-regular \PDL{}-like specification logic is easily obtained by extending the Kleene algebra of accessibility
      relations, resp.\ the class of regular languages in modal operators, to larger language classes like the context-free ones, 
      resulting in \PDLCFL \cite{JCSS::HarelPS1983}. It can state properties like ``\emph{there is a path of the form $a^nb^n$ for
      some $n \ge 1$}''.
\item Fixpoint Logic with Chop (\FLC) \cite{Mueller-Olm:1999:MFL} extends \mucalc with an operator for sequential composition. This 
      enables it to express context-free properties like the only mentioned above for \PDLCFL but also other non-regular ones like
      ``\emph{all paths end after the same number of steps}.'' In fact, \FLC embeds \PDLCFL \cite{langesomla-ipl06}.
\item Assume-guarantee properties like ``\emph{every execution of program $P$ by $n$ steps, can be matched by an execution of program
      $Q$ by $n+1$ steps}'' are not-regular and have led to invention of \emph{Higher-Order Fixpoint Logic} (\HFL{}) \cite{viswan:hfl:2004}, 
      an extension of \mucalc by a typed $\lambda$-calculus. This captures \FLC on a very low type level and stretches far beyond that.
\end{itemize} 
A common feature of such extensions -- at least when not done carefully -- is the loss of decidability of the satisfiability problem: 
\PDLCFL is highly undecidable \cite{JCSS::HarelPS1983}, and this transfers to \FLC and \HFL{}. On the other hand, their model checking
problems over finite LTS remains decidable, making them suitable for automatic program verification. For \PDLCFL it is even polynomial
\cite{lange-mcpdl}, whereas for \FLC is it \EXPTIME-complete \cite{lange:3mcflc:06}, and for general \HFL{} it is non-elementary 
\cite{als-mchfl07}.  

\paragraph*{Decidable Non-Regular Program Logics.}
A discovery in formal language theory has opened up some interesting possibilities in the realm of non-regular program logics: the
class of \emph{visibly pushdown languages} (\VPL) over some visibly pushdown alphabet partitioning alphabet symbols into those that 
cause push-, pop- and internal state changes in a corresponding pushdown automaton, forms a subset of \CFL\ that enjoys almost the 
same closure and decidability properties as the class \REG\ \cite{STOC04:202}. This is even robust in the sense that the corresponding
\PDL{} over this class, \PDLVPL, is a genuinely non-regular program logic whose satisfiability problem is actually decidable, namely
\TwoEXPTIME-complete \cite{journals/jlp/LodingLS07}.

This comprised the state-of-the-art until recently: the largest (w.r.t.\ expressiveness) explicit-fixpoint logic that was still 
known to be decidable was \mucalc, as known extensions thereof were undecidable. Regarding implicit-fixpoint logics in the form of
propositional dynamic ones, it was possible to push the decidability frontier genuinely into the realm of non-regular logics by
the definition of \PDLVPL. A similar construction was possible for a temporal logic called \emph{Recursive \CTL} (\RecCTL) 
\cite{conf/time/BruseL20} which can be seen as a \CTL-like variant of \FLC: it was possible to define a (non-regular) fragment 
which could be shown to be decidable by a simple satisfiability-preserving translation into \PDLVPL \cite{conf/time/BruseL20}.

A common feature of these decidable logics is the modular use of visibly-pushdown effects. This is best seen in \PDLVPL where
\VPL\ can either be used inside an existential or a universal modality, but no mixing of modalities across visibly pushdown languages
is possible. Hence, \PDLVPL can express properties like $\mudiam{a^nb^n}p$, but cannot state ``\emph{there is an $a$-path of length
$n$ for some $n$, such that all following $b$-paths of length $n$ for the same $n$ end in a state satisfying $p$}'' because this 
requires a change of modality \emph{within} the \VPL\ $a^nb^n$. We will write $\mudiam{a^n}\mubox{b^n}$ to denote this property 
succinctly, even though this is no well-formed formula of any (\PDL{}-like) logic.

Since decidability of \VPL-based logics can be obtained by a reduction to visibly pushdown games \cite{conf/fsttcs/LodingMS04}, there
is little reason to believe that the strict use of \VPL within a single modality is actually required for decidability. In fact,
very recently a fragment of \FLC, called \vpFLC, has been constructed which allows free use of modality changes within \VPL{}s, and
whose satisfiability problem is still decidable \cite{conf/concur/BruseL21}.

\paragraph*{Contribution.} 
The aim of this paper is to investigate the expressive power of program logics, specifically those around the decidability border,
and to show that the existing inclusions between them are strict. This is already known for a few of them, either because of 
long-standing results like the strict inclusions between regular logics and their non-regular extensions, or because of rather
trivial observations. Note for instance that it has been known for a long time that \PDLCFL and \mucalc are incomparable w.r.t.\
expressiveness. Hence, \FLC, which subsumes them both, must subsume both of them strictly. This is fair enough, but also slightly
non-satisfactory, when one considers examples witnessing the strictness that result from this kind of reasoning.
\begin{itemize}
\item \PDLCFL is designed to express some non-regular properties, whereas \mucalc can only express regular ones. Hence, any genuinely
      non-regular property in \FLC witnesses the strictness of the inclusion of \mucalc in \FLC.
\item On the other hand, \PDL{}-based logics cannot express properties that involve unbounded modality alternation like 
      $(\mudiam{a}\mubox{b})^n := \mu X. p \vee \mudiam{a}\mubox{b} X$ (defining the winning region for one of the players in a turn-based 
      two-player game). Hence, any such formula also witnesses the strict inclusion of \PDLCFL in \FLC. This is not very satisfactory, 
      as this does not shed light onto the true non-regular difference of these genuinely non-regular program logics.  
\end{itemize} 
In this paper we complete the study into the (non-expressiveness) of properties expressible in program logics around the decidability
border. The hierarchy formed by them is shown in Fig.~\ref{fig:hierarchy}, naming specific properties witnessing the strictness of the
corresponding inclusion as well as pointing to their origin. The picture is made complete by results in this paper, providing new 
separation results as well as tighter separation results in the form of witnessing properties which 
genuinely do not rely on fairly trivial inexpressibility results between the regular and non-regular world. To denote these properties
we use a fairly intuitive notation like $\mudiam{a^nb^n}$ etc. In the end, we can see that all the inclusions in Fig.~\ref{fig:hierarchy}
are strict. In particular,
\begin{itemize}
\item \ldots propositional dynamic logics and modal fixpoint logics can be separated using using properties of unbounded modality
      alternation (e.g.\ $(\mudiam{a}\mubox{b})^n$) or, further up in the hierarchy, those in which modal alternation and limited
      counting is intertwined (e.g.\ $\mudiam{a^n}\mubox{b^n}$ or $\mudiam{a^n}\mubox{b}\mudiam{a^n}$). As a result, \vpFLC is not only strictly more expressive than 
      \PDLVPL and \mucalc (which was known before), but even strictly more expressive than their union.  
\item \ldots the visibly pushdown based logics (middle band) can be separated from the general non-regular ones (top band) using properties
      based on visibly pushdown languages. This may not sound surprising but is not a triviality as program properties using 
      non-visibly-pushdown languages could, in theory, be composable by complex formulas using visibly pushdown 
      languages only. Here we show that this is indeed not the case.
\end{itemize}

\begin{figure}
\begin{center}
\begin{tikzpicture}[semithick, node distance=25mm]
  \tikzstyle{logic}=[shape=rectangle,rounded corners,semithick,draw,fill=gray!20]
  
  \node[logic] (pdl)                      {\PDL{}};
  \node[logic] (mucalc) at ++(5,1)        {\mucalc};
  \node[logic] (pdlvpl) [above of=pdl]    {\PDLVPL};
  \node[logic] (pdlcfl) [above of=pdlvpl] {\PDLCFL};
  \node[logic] (vpflc)  [above of=mucalc] {\vpFLC};
  \node[logic] (flc)    [above of=vpflc]  {\FLC};
  
  \path (pdl) edge[dashed] node [above,sloped] {$(\mudiam{a}\mubox{b})^n$} node [below,sloped] {\cite{Kozen83}} (mucalc)
              edge[dashed] node [left,pos=.4]  {$\mudiam{a^nb^n}$}  node [right,pos=.4] {\cite{JCSS::HarelPS1983}} (pdlvpl)
        (pdlvpl) edge[dashed] node [above,sloped] {$\mudiam{a^n}\mubox{b^n}$} node [below,sloped] {Thm.~\ref{thm:diabox}} (vpflc)
                 edge[dashed] node [left,pos=.4]  {$\mudiam{a^nba^n}$}  node [right,pos=.4] {Thm.~\ref{thm:sepundec}} (pdlcfl)
        (pdlcfl) edge[dashed] node [above,sloped] {$\mudiam{a^n}\mubox{b}\mudiam{a^n}$} node [below,sloped] {Thm.~\ref{thm:an-b-an}} (flc)
        (mucalc) edge[dashed] node [left,pos=.5] {$\mudiam{a^nb^n}$} node [right,pos=.5] {\cite{Mueller-Olm:1999:MFL}} (vpflc)
        (vpflc) edge[dashed] node [left,pos=.5] {$\mudiam{a^nba^n}$} node [right,pos=.5] {Thm.~\ref{thm:sepundec}} (flc);
        
  \path[dotted] (-2,1.7) edge (10,1.7)
                (-2,4.2) edge (10,4.2); 
                 
  \node[anchor=west] at (6.5,1.2) {regular expressiveness, decidability};  
  \node[anchor=west] at (6.5,3.7) {non-regular expressiveness, decidability};  
  \node[anchor=west] at (6.5,6) {non-regular expressiveness, undecidability};  
\end{tikzpicture}
\end{center}
\caption{The (strict) hierarchy of propositional dynamic and modal fixpoint logics betwen \PDL{} and \FLC.} 
\label{fig:hierarchy}
\end{figure}
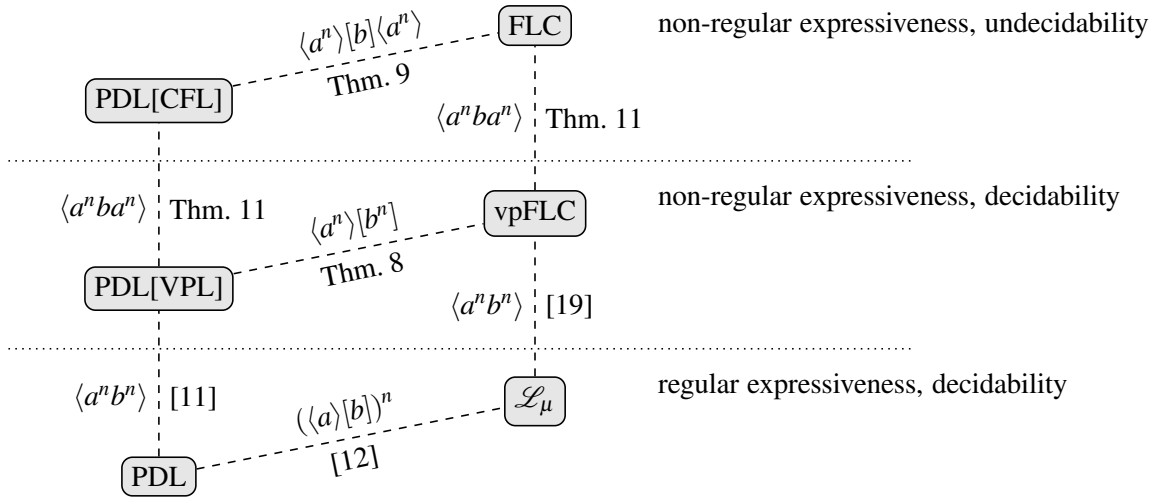

\paragraph*{Organisation.} In Sect.~\ref{sec:prel} we recall necessary preliminaries from formal languages, propositional dynamic and 
modal fixpoint logics. In Sect.~\ref{sec:sepmfpfrompdl} we provide the missing results that separate the propositional dynamic logics
from the modal fixpoint logics, i.e.\ the left column from the right column in Fig.~\ref{fig:hierarchy}. This concerns the non-regular
parts, as the separation of \mucalc from \PDL{} is well-known.

In Sect.~\ref{sec:sepdecfromundec} we then show how to leverage the undecidability of the satisfiability problem for such logics into
an expressiveness gap, thus separating the top row from the middle row in Fig.~\ref{fig:hierarchy}. Again, the separation of the middle
row from the bottom row was known already, with proofs relying for instance on the finite model properties of the regular logics
\PDL{} and \mucalc. We conclude with remarks on further work in Sect.~

\section{Preliminaries}
\label{sec:prel}
% !TEX root =  main.tex

\subsection{Languages}

An \emph{alphabet} is a finite, nonempty set of \emph{letters}. A \emph{word} over some 
alphabet $\Sigma$ is a finite sequence $w = w_1\dotsb w_n$ of letters from $\Sigma$. The
 empty word is denoted by $\eword$. The length $|w|$ of $w = w_1\dotsb w_n$ is $n$. The set
 of all $\Sigma$-words is denoted by $\Sigma^*$, the set of all non-empty words is denoted 
by $\Sigma^+$. A \emph{$\Sigma$-language} is a subset of $\Sigma^*$. We denote concatenation
 of words and languages by juxtaposition as usual. 

A \emph{visibly pushdown alphabet} is an alphabet $\Sigma$ that is partitioned into three 
sets $\Sigma = \Sigma_{\call} \cup \Sigma_{\ret} \cup \Sigma_{\internal}$ of \emph{call}, 
\emph{return} and \emph{internal} symbols.

\paragraph*{Finite Automata.}

A (nondeterministic) \emph{finite automaton} (NFA) is a tuple $(Q, \Sigma, \delta, q_I, F)$
 where $Q$ is a finite, nonempty set of states, $\Sigma$ is an alphabet, $q_I \in Q$ is the
 \emph{initial state}, $F \subseteq Q$ is the set of accepting states and 
$\delta \subseteq Q \times \Sigma \times Q$ is the transition relation.

A \emph{run} of some NFA $\aut{A}$ on a word $w = a_1\dotsb a_n \in \Sigma^*$ is a sequence 
$q_0,\dotsc,q_n \in Q^*$ such that, for all $0 \leq i < n$, we have that $(q_i, a_i, q_{i+1}) \in \delta$. 
We also say that $\aut{A}$ has a run from $q_0$ to $q_n$ over $w$. Such a run is called accepting if 
$q_0 = q_I$ and $q_n \in F$. A word $w$ is accepted by $\aut{A}$ if there is an accepting run of 
$\aut{A}$ on $w$. We write $L(\aut{A})$ for the set of words accepted by $\aut{A}$. A finite 
automaton $(Q, \Sigma, \delta, q_i, F)$ is called \emph{deterministic}, if for all $q \in Q, a \in \Sigma$, 
there is exactly one $q' \in Q$ such that $(q, a, q') \in \delta$. It is well-known that, for each NFA, 
there is a DFA that accepts exactly the same language, and has at most exponentially more states. A language 
is called \emph{regular} if there is a finite automaton that accepts it. We write $\REG$ for the class of regular 
languages (over a given alphabet).

\paragraph*{Pushdown Automata.}

A \emph{pushdown automaton} is a tuple $(Q, \Sigma, \Gamma, \stdown, \delta, q_I, F)$ where $Q, \Sigma, q_I, F$
 are as in the case of finite automata,  $\Gamma$ is the \emph{stack alphabet} such that 
$\Gamma \cap \Sigma = \emptyset$, $\stdown \notin \Gamma \cup \Sigma$ is the stack bottom symbol,  
$\delta \subseteq (Q \times \Sigma \times \Gamma \times \Gamma^* \times Q) \cup 
(Q \times \Sigma \times \{\stdown\} \times \Gamma^* \times Q)$ is the transition relation. 

A \emph{run} of some pushdown automaton (PDA) $\mathcal{A}$ on a word $w = a_1\dotsb a_n \in \Sigma^*$ is 
a finite sequence $(q_0, \gamma_0),\dotsc,(q_n, \gamma_n) \in (Q \times \Gamma^*)^*$ of states and stack 
contents such that, for all $0 \leq i < n$ either $\gamma_i = \eword$ and $(q_i, a_i, \stdown, \gamma_{i+1}, q_{i+1}) 
\in \delta$ or $\gamma_i = \gamma'_i \gamma'$ with $|\gamma'| = 1$ and $\gamma_{i+1} = \gamma'_i \gamma''$ 
and $(q_i, a_i, \gamma', \gamma'', q_{i+1}) \in \delta$. A run is \emph{accepting} if $q_0 = q_I, \gamma_0 = \eword$ 
and $q_n \in F$. A word $w$ is accepted by $\mathcal{A}$ if there is an accepting run of $\mathcal{A}$ on $w$. 
We write $L(\mathcal{A})$ for the language of words accepted by $\mathcal{A}$. A language is called \emph{context-free}
 if it is accepted by some PDA. We write $\CFL$ for the class of context-free languages. Note that every NFA can be extended 
to a PDA by ignoring the stack, hence every regular language is context-free. Also, context-free languages over a unary 
(size $1$) alphabet are known to be regular \cite{JACM::Parikh66}.

A \emph{visibly pushdown automaton} (VPA) \cite{Mehlhorn/80,STOC04:202} is a pushdown automaton such that its alphabet 
$\Sigma$ is a visibly pushdown alphabet and, moreover, for $(q, a, \gamma, \gamma', q') \in \delta$, we have that either 
\begin{itemize}
\item $a \in \Sigma_{\call}$ and $|\gamma'| = 2$,
\item $a \in \Sigma_{\ret}$, $\gamma \not = \stdown$, and $|\gamma'| = 0$,
\item $a \in \Sigma_{\internal}$ and $|\gamma'| = 1$,
\end{itemize}
and, finally, a run is only accepting if the final configuration has an empty stack.
A language over some visibly pushdown alphabet is called \emph{visibly pushdown} if it is accepted by some VPA. We
write $\VPL$ for the class of visibly pushdown languages (over a given visibly pushdown alphabet), Clearly 
every visibly pushdown language is context-free. Moreover, every regular language is visibly-pushdown over the alphabet
that regards all letters as internal.

Note that this definition introduces visibly pushdown languages without so-called pending calls and returns, i.e.\ 
such that only words are accepted that are balanced and well-nested w.r.t.\ symbols from $\Sigma_{\call}$ and $\Sigma_{\ret}$. 
This is done to make the definition of \vpFLC (see Sec.~\ref{sec:flc} below) more accessible.

\paragraph*{Derivatives.}

Let $L$ be a $\Progs$-language and let $a \in \Progs$. The $a$-derivative of $L$ is the set 
$\deriv{a}{L} := \{w \in \Progs^* \mid aw \in L\}$ \cite{Rabin:Scott:ibmjrd:1959}.
It is easy to see that derivatives of regular languages are regular again by manipulating the associated 
deterministic finite automaton, i.e.\ making the $a$-successor of the initial state the new initial state, cf.\ 
\cite{JACM::Brzozowski64}. 

We now argue that the $a$-derivative of a context-free language is context-free again. 

\begin{lemma}
%\label{obs:derivatives}
\label{lem:derivatives}
Let $L \in \CFL$ over some alphabet $\Sigma$ and $a \in \Sigma$. Then $\deriv{a}{L} \in \CFL$.
\end{lemma}

\begin{proof}
Suppose $\aut{A} = (Q, \Sigma, \Gamma, \stdown, \delta, q_I, F)$ is a PDA with $L(\aut{A}) = L$. Let 
$T = \{(q_i, a, \stdown, \gamma, q) \in \delta \mid \gamma \in \Gamma^*, q \in Q \}$ be the set of transitions 
available to $\aut{A}$ upon reaching an $a$ as the first letter of a word. For each $t =(q_i, a, \stdown, \gamma, q)  \in T$,
if $\gamma = \eword$, define the automaton $\aut{A}_t = (Q \cup \{q'\}, \Sigma, \Gamma, \stdown, \delta_t, q', F)$ with 
\[
\delta_t = \delta \cup \{ (q', b, \stdown, \gamma', q'') \mid (q', b, \stdown, \gamma', q'') \in \delta \}
\]
and if $\gamma = \gamma' \gamma''$ with $|\gamma''| = 1$, 
define the automaton $\aut{A}_t = (Q \cup \{q'\}, \Sigma, \Gamma, \stdown, \delta_t, q', F)$ with 
\[
\delta_t = \delta \cup \{ (q', b, \stdown, \gamma'\gamma''', q'') \mid (q', b, \gamma'', \gamma''', q'') \in \delta\}. 
\]
Intuitively, $\aut{A}_t$ simulates the case of $\aut{A}$ reading an $a$ and taking $t$ as its first transition,
and then proceeds like $\aut{A}$ would on the rest of the word. Hence, $\aut{A}_t$ accepts only words in $\deriv{a}{L(\aut{A})}$. 

Conversely, each word in $\deriv{a}{L(\aut{A})}$ is accepted by at least one of the
$\aut{A}_t$. Using the well-known closure of context-free languages under union, we obtain the desired statement. 
\end{proof}

\subsection{Propositional Dynamic Logic}

\paragraph*{Labelled Transition Systems.}

Let $\Prop$ be a set of atomic propositions and let $\Progs$ be an alphabet, in the context of propositional dynamic logics often 
referred to as the set of \emph{atomic programs}. 

A \emph{labelled transition system} 
(LTS) is a tuple $\Transsys = (\States, \Transition{}{}{}, \Label)$ where $\States$ is a, potentially infinite, set of 
\emph{states}, $\Transition{}{}{} \subseteq \States \times \Progs \times \States$ is the transition relation, and 
$\Label \colon \States \to 2^\Prop$ labels each state with the set of propositions that hold at it. Given an LTS, we write
 $\Transition{s}{a}{t}$ to denote that $(s, a, t) \in \Transition{}{}{}$. This extends to $\Progs$-words via 
$\Transition{s}{w_1 w_2}{t}$ iff there is $s'$ with $\Transition{s}{w_1}{s'}$ and $\Transition{s'}{w_2}{t}$, and 
$\Transition{s}{\eword}{t}$ iff $s = t$. We write $\Transition{s}{L}{t}$ if there is $w \in L$ such that $\Transition{s}{w}{t}$.
Sometimes we consider labelled transition systems with a designated initial state. In a drawing, this will be marked by an 
ingoing edge with no source. 

A (finite) path $\pi$ in an LTS is a sequence $s_1,\dotsc,s_n$ of states such that, for all $1 \leq i < n$, there is 
$a_i \in \Progs$ such that $\Transition{s_i}{a_i}{s_{i+1}}$. In this case, the path is said to be labelled by 
$w = a_1\dotsb a_{n-1}$. Note that $\Transition{s}{w}{t}$ iff there is a $w$-labelled path from $s$ to $t$.

\paragraph{Syntax.}

Let $\Prop$ and $\Progs$ be as before. Let $\lclass$ be a, not necessarily finite, set of $\Progs$-languages. The syntax
 of Propositional Dynamic Logic over $\lclass$ ($\PDL{\lclass}$) is defined by the following grammar:
\[
\varphi \quad ::= \quad p \mid  \neg \varphi \mid \varphi \vee \varphi \mid \varphi  \wedge \varphi  \mid \mudiam{L}\varphi \mid \mubox{L} \varphi
\]
where $p \in \Prop$ and $L \in \lclass$. The auxiliary formulas $\mutrue$ and $\mufalse$ are defined as usual via $p \vee \neg p$, resp.\ $p \wedge \neg p$
for arbitrary $p \in \Prop$. We are particularly interested in those logics in which $\lclass = \REG, \VPL$ or $\CFL$,
resulting in the logics 
\begin{itemize}
\item \emph{Propositional Dynamic Logic of Regular Programs} (\PDLREG) \cite{Fischer79}, sometimes only called \PDL{}, 
\item \emph{Propositional Dynamic Logic of Context-Free Programs}\footnote{Originally it was called 
      \emph{Propositional Dynamic Logic of Non-Regular Programs} which is of course slightly misleading as it hardly comprises
      \emph{all} non-regular programs.} (\PDLCFL) \cite{JCSS::HarelPS1983}, and
\item \emph{Propositional Dynamic Logic of Recursive Programs} (\PDLVPL) \cite{journals/jlp/LodingLS07}.
\end{itemize}

The \emph{modal depth} $\md{\varphi}$ of a formula $\varphi$ measures the nesting depth of modal operators and is defined inductively via
\begin{align*}
\md{p}  &= 0 \\
\md{\mudiam{L}\varphi} = \md{\mubox{L}\varphi} = \md{\neg \varphi} &= 1+ \md{\varphi} \\
\md{\varphi_1 \vee \varphi_2} = \md{\varphi_1 \wedge \varphi_2} &= \max(\md{\varphi_1}, \md{\varphi_2})
\end{align*}

\paragraph*{Semantics.}

Given an LTS $\Transsys = (\States, \Transition{}{}{}, \Label)$, any $\PDL{\lclass}$-formula over matching $\Prop$ and 
$\Progs$ defines a subset $\sem{\varphi}{\Transsys}{}$ of $\States$, inductively defined via
\begin{align*}
\sem{p}{\Transsys}{} &= \{s \in \States \mid p \in \Label(s)\} & \sem{\neg \varphi}{\Transsys}{} &= \States \setminus \sem{\varphi}{\Transsys}{} \\
\sem{\varphi_1 \vee \varphi_2}{\Transsys}{} &= \sem{\varphi_1}{\Transsys}{} \cup \sem{\varphi_2}{\Transsys}{} &  
\sem{\varphi_1 \wedge \varphi_2}{\Transsys}{} &= \sem{\varphi_1}{\Transsys}{} \cup \sem{\varphi_2}{\Transsys}{} \\
\sem{\mudiam{L}\varphi}{\Transsys}{} &= \{s \mid \text{ ex. } t \in \sem{\varphi}{\Transsys}{}, \text{ s.t. } \Transition{s}{L}{t}\} & 
\sem{\mubox{L}\varphi}{\Transsys}{} &= \{s \mid \text{ f.a. } t \text{ s.t. } \Transition{s}{L}{t}, \text{ then } t \in \sem{\varphi}{\Transsys}{} \}.
\end{align*}

We write $\Transsys, s \models \varphi$ to denote that $s \in \sem{\varphi}{\Transsys}{}$. We say that $\Transsys$ 
is a model of $\varphi$ to denote that the initial state of $\Transsys$ is in $\sem{\varphi}{\Transsys}{}$.

\subsection{Fixpoint Logic with Chop}
\label{sec:flc}

We assume some familiarity with the modal $\mu$-calculus \mucalc.
The yardstick in the world of modal fixpoint logics to measure the expressive power of propositional dynamic logics against has
turned out to be Fixpoint Logic with Chop \cite{Mueller-Olm:1999:MFL}, the extension of \mucalc by a sequential composition operator. 
This is since some kind of sequential composition operator is needed in many language classes beyond the regular 
languages, and \FLC is, in some sense, a minimal extension of \mucalc by such an operator.
As \FLC is by far less known than \mucalc, we include its definition here. It is also needed to explain its fragment \vpFLC, comprising
the largest currently known modal fixpoint logic with a decidable satisfiability problem \cite{conf/concur/BruseL21}.

\paragraph*{Syntax.}
Let $\Prop$ and $\Progs$ be as above. Let $\Var$ be a countably infinite set of variable names. Formulas of 
\FLC over $\Prop$, $\Progs$ and $\Var$ are given by the following grammar.
\begin{displaymath}
\varphi \quad ::= \quad q \mid \overline{q} \mid X \mid \term \mid \mudiam{a} \mid \mubox{a} \mid \varphi \vee \varphi \mid 
\varphi \wedge \varphi \mid \mu X.\varphi \mid \nu X.\varphi \mid \varphi ; \varphi 
\end{displaymath}
where $q \in \Prop$, $a \in \Progs$ and $X \in \Var$. Note that \FLC does not have a negation operator.

\paragraph*{Semantics.}

Let $\Transsys = (\States,\Transition{}{}{},\Label)$ be an LTS. An \emph{environment} $\eta: \Var \to (2^\States \to 2^\States)$ 
assigns to each variable a function from sets of states to sets of states in $\Transsys$. We write $\eta[X \mapsto f]$ for the 
function that maps $X$ to $f$ and agrees with $\eta$ on all other arguments. 

The semantics $\sem{\cdot}{\eta}{\Transsys}: 2^\States \to 2^\States$ of an \FLC formula, relative to an LTS $\Transsys$ and
an environment, is such a function. It is monotone with respect to the inclusion ordering on 
$2^\States$. Such functions together with the partial order given by
\begin{displaymath}
f \sqsubseteq g \quad \mbox{iff} \quad \forall T \subseteq \States: f(T) \subseteq g(T)
\end{displaymath}
form a complete lattice with joins $\sqcup$ and meets $\sqcap$ -- defined as the pointwise union, resp.\ intersection. By the 
Knaster-Tarski Theorem \cite{Tars55} the least and greatest fixpoints of monotone functionals 
$F: (2^\States \to 2^\States) \to (2^\States \to 2^\States)$ 
exist. They are used to interpret fixpoint formulas of \FLC. The semantics is then inductively defined as follows.
\begin{align*} 
\sem{q}{\Transsys}{\eta} & = \underline{\enspace} \mapsto \{ s \in \States \mid q \in \Label(s) \} &
\sem{\mudiam{a}}{\Transsys}{\eta} & = T \mapsto \{ s \in \States \mid \exists t \in T \mbox{ s.t. } \Transition{s}{a}{t} \} \\
\sem{\overline{q}}{\Transsys}{\eta} & = \underline{\enspace} \mapsto \{ s \in \States \mid q \not\in \Label(s) \} &
\sem{\mubox{a}}{\Transsys}{\eta} & = T \mapsto \{ s \in \States \mid \forall t \in T: \Transition{s}{a}{t} \Rightarrow t \in T \} \\
\sem{X}{\Transsys}{\eta} & = \eta(Z) &
\sem{\mu X.\varphi}{\Transsys}{\eta} & = \bigsqcap \{ f: 2^\States \to 2^\States \mid f \mbox{ mon., }
                        \sem{\varphi}{\Transsys}{\eta[X \mapsto f]} \sqsubseteq f \} \\
\sem{\term}{\Transsys}{\eta} & = T \mapsto T &
\sem{\nu X.\varphi}{\Transsys}{\eta} & = \bigsqcup \{ f: 2^\States \to 2^\States \mid f \mbox{ mon., }
                        f \sqsubseteq \sem{\varphi}{\Transsys}{\eta[X \mapsto f]} \} \\
\hspace*{-3mm}\sem{\varphi ; \psi}{\eta}{\Transsys} & = \sem{\varphi}{\Transsys}{\eta} \circ \sem{\psi}{\Transsys}{\eta} &
\sem{\varphi \vee \psi}{\Transsys}{\eta} & = \sem{\varphi}{\Transsys}{\eta} \sqcup \sem{\psi}{\Transsys}{\eta} \\
&& \sem{\varphi \wedge \psi}{\Transsys}{\eta} & = \sem{\varphi}{\Transsys}{\eta} \sqcap \sem{\psi}{\Transsys}{\eta}
\end{align*}
Here, $\circ$ denotes function composition.

For any \FLC formula $\varphi$, any LTS $\Transsys = (\States,\Transition{}{}{},\Label)$  with initial state $s_0$ and any environment $\eta$ let 
$\Sem{\varphi}{\eta}{\Transsys} := \sem{\varphi}{\eta}{\Transsys}(\States)$. 
We call this the set of positions in $t$ \emph{defined} by $\varphi$ and $\eta$. 
We also write $\Transsys,s \models_\eta \varphi$ if
$s \in \Sem{\varphi}{\eta}{\Transsys}$, resp.\ $\Transsys \models_\eta \varphi$ if $\Transsys,s_0 \models_\eta \varphi$. If 
$\varphi$ is closed we may omit $\eta$ in both kinds of notation. We say that $\Transsys$ is a \emph{model} of a closed formula
$\varphi$ if $\Transsys \models \varphi$. A formula is \emph{satisfiable} if it has a model.

Two formulas $\varphi$ and $\psi$ are \emph{equivalent}, written $\varphi \equiv \psi$, iff their semantics are the same, i.e.\ 
for every environment $\eta$ and every LTS $\Transsys$: $\sem{\varphi}{\eta}{\Transsys} = \sem{\psi}{\eta}{\Transsys}$. Two formulas 
$\varphi$ and $\psi$ are \emph{weakly equivalent}, written $\varphi \approx \psi$, iff they define the same set of states
in an LTS, i.e.\ for every $\eta$ and every $\Transsys$ we have $\Sem{\varphi}{\eta}{\Transsys} = \Sem{\psi}{\eta}{\Transsys}$.
Hence, we have $\varphi \approx \varphi;\mutrue$ for any $\varphi$. 

\paragraph*{Visibly Pushdown \FLC.}

\begin{definition}
Let $\Sigma = \Sigma_{\call} \cup \Sigma_{\ret} \cup \Sigma_{\internal}$ be a visibly pushdown alphabet.
The syntax of the fragment \vpFLC of \FLC is given by the following grammar.
\begin{align*}
\varphi \quad ::= \quad &q \mid \overline{q} \mid X \mid \psi \vee \psi \mid \psi \wedge \psi \mid \mu X.\psi \mid \nu X.\psi \mid \\
 &\mudiabox{a_\internal} \mid \mudiabox{a_\internal};\varphi 
 \mid \mudiabox{a_\call};\mudiabox{a_\ret} \mid \mudiabox{a_\call};\varphi;\mudiabox{a_\ret} \mid 
   \mudiabox{a_\call};\mudiabox{a_\ret};\varphi \mid \mudiabox{a_\call};\varphi;\mudiabox{a_\ret};\varphi
\end{align*}
where $q \in \Prop$, $X \in \Var$, $a_m \in \Act_x$ for $m \in \{\internal,\call,\ret\}$, and $\mudiabox{a}$ can be either 
$\mudiam{a}$ or $\mubox{a}$. Furthermore, we postulate that the sequential composition operator is right-associative; parentheses
are not shown explicitly here for the sake of better readability. 

The definition alongside the visibly pushdown alphabet ensures that sequential composition, in particular when involving multiple composition operators,
 appears only ``guarded'' by modal operators. 
Hence, when exploring several formulas at the same time, e.g.\ in a tableau or a satisfiability game (cf.~\cite{conf/concur/BruseL21}), unfolding
modal operators in lockstep will ensure that these formulas always have similar amounts of nested chop operators, and, hence, are of similar size.
This makes such a satisfiability game a stair-parity game and, hence decidable.

It is open whether the definition of \vpFLC can be relaxed to allow pending calls and returns.

\end{definition}

\subsection{Separating Properties}
\label{sec:sepprop}

The properties witnessing the separation results discussed in this paper have -- to some degree -- been mentioned in the 
introduction already, and are also shown in Fig.~\ref{fig:hierarchy}. Here we defined them formally as formulas of the 
corresponding logics.
\begin{proplist}
\item $\bm{\mudiam{a^nb^n}}$. This property simply states ``\emph{there is a path labelled with $n$ $a$'s, followed by
      $n$ $b$'s, for some $n \ge 1$, ending in a state where $p$ holds}''. It should be clear that this can be expressed in \PDLCFL since 
      $L_{a^nb^n} := \{ a^nb^n \mid n \ge 1 \}$ is a \CFL. Hence, $\mudiam{a^nb^n} := \mudiam{L_{a^nb^n}}p$ is a formula doing so.
      
      Note that $L_{a^nb^n}$ is in fact a \VPL\ over the visibly pushdown alphabet with $a$ being a call and $b$ a return symbol.
      Hence $\mudiam{a^n b^n} \in \PDLVPL$.
      
      An \FLC formula formalising this property is $(\mu Z.\mudiam{a};\mudiam{b} \vee \mudiam{a};Z;\mudiam{b});p$. It is best 
      understood by unfolding the fixpoint formula using
      \begin{align*}
        \underbrace{(\mu Z.\mudiam{a};\mudiam{b} \vee \mudiam{a};Z;\mudiam{b})}_{\psi_Z};p 
        &\equiv (\mudiam{a};\mudiam{b} \vee \mudiam{a};\psi_Z;\mudiam{b});p 
         \equiv \mudiam{a};\mudiam{b};p \vee \mudiam{a};\psi_Z;\mudiam{b};p \\
        &\equiv \mudiam{a};\mudiam{b};p \vee \mudiam{a};(\mudiam{a};\mudiam{b} \vee \mudiam{a};\psi_Z;\mudiam{b});\mudiam{b};p \\
        &\equiv \mudiam{a};\mudiam{b};p \vee \mudiam{a};\mudiam{a};\mudiam{b};\mudiam{b};p \vee \mudiam{a};\mudiam{a};\psi_Z;\mudiam{b};\mudiam{b};p \\
        &\enspace \vdots \\
        &\equiv \bigvee\limits_{n \ge 1} \underbrace{\mudiam{a};\ldots,\mudiam{a}}_{n \text{ times}};\underbrace{\mudiam{b};\ldots;\mudiam{b}}_{n \text{ times}};p.
      \end{align*}
			Clearly, this property cannot be expressed in \PDLREG \cite{JCSS::HarelPS1983}.

\item $\bm{\mudiam{a^n b a^n}}$. This is very similar to the previous property with the \CFL\ $L_{a^nba^n} = \{ a^nba^n \mid n \ge 1 \}$ instead, hence 
      $\mudiam{a^n b a^n}$ it is a \PDLCFL property. By the standard translation into \FLC, it is also expressible there, for instance
      as $\varphi_{a^nba^n} := (\mu Z.\mudiam{a};\mudiam{b};\mudiam{a} \vee \mudiam{a};Z;\mudiam{a});p$. Using fixpoint unfolding as before, one can
      see that it is equivalent to $\bigvee_{n \ge 1} \mudiam{a}^n;\mudiam{b};\mudiam{a}^n;p$.
      
      Note that $L_{a^nba^n}$ is not a \VPL\ since $a$ would need to be both a call and a return symbol in order to be recognisable 
      using a pushdown automaton. Likewise, $\varphi_{a^nba^n}$ is not a \vpFLC formula as $\mudiam{a}$ occurs both in front of and
      behind a recursion variable. Again, this would require $a$ to be both a call and a return symbol. We formally show this in Thm.~\ref{thm:sepundec}.
       
\item $\bm{(\mudiam{a}\mubox{b})^n}$. This property is supposed to state something occurring in the definition of winning regions
      in alternating two-player reachability games of unbounded iteration, namely that there is some number $n \ge 1$ of moves such that player 1 can make a 
      move such that no matter how player 2 responds, player 1 can make another move, etc.\ until after $2n$ moves a states satisfying
      $p$ is reached. This is expressed by the \mucalc formula $\mu X.\mudiam{a}\mubox{b}(p \vee X)$ or, using the standard translation
      into \FLC \cite{Mueller-Olm:1999:MFL}, as $(\mu X.\mudiam{a};\mubox{b}(p \vee X));\mutrue$ which is also a \vpFLC formula. 
			However, this property cannot be expressed in \PDLREG \cite{Kozen83}.
      
\item $\bm{\mudiam{a^n}\mubox{b^n}}$. This is similar to the property $\mudiam{a^nb^n}$ but here the second part ``of the path'' is
      universally quantified, i.e.\ it asks for the existing of an $a$-path of some length $n \ge 1$ such that all $b$-paths of length
      $n$ following it end in a state satisfying $p$. This is easily expressed in \FLC by changing the corresponding $\mudiam{b}$ to
      a $\mubox{b}$ in the formula for $\mudiam{a^nb^n}$, resulting in $(\mu Z.\mudiam{a};\mubox{b} \vee \mudiam{a};Z;\mubox{b});p$. 
      Since the original formula was already \vpFLC, so is this one, as \vpFLC treats existential and universal modalities equally.
			However, we show in Thm.~\ref{thm:diabox} that this property cannot be expressed in \PDLCFL.

\item $\bm{\mudiam{a^n}\mubox{b}\mudiam{a^n}}$. This asks for the existence of an $a$ path of length $n$, such that all $b$-successors
      of the target state have an emerging $a$-path of length $n$ again to some state satisfying $p$. Again, an \FLC formula for this
      property is easily obtained from one for $\mudiam{a^n b a^n}$ by changing a corresponding modality, resulting in 
      $(\mu Z.\mudiam{a};\mubox{b};\mudiam{a} \vee \mudiam{a};Z;\mudiam{a});p$. This is also no \vpFLC formula for the same reason that $a$ cannot have
      two roles in a visibly pushdown alphabet. We use this property to separate \PDLCFL and \FLC in Thm.~\ref{thm:an-b-an} by showing that it
			cannot be expressed in the former.
\end{proplist}

\section{Separating Propositional Dynamic and Modal Fixpoint Logics}
\label{sec:sepmfpfrompdl}

% !TEX root =  main.tex

In this section we prove that particular properties separate the expressive power of the modal fixpoint logics in the
middle and upper band of Fig.~\ref{fig:hierarchy} from the propositional dynamic logics in these bands, namely the
properties $\mudiam{a^n}\mubox{b}^n$ and $\mudiam{a^n}\mubox{b}\mudiam{a^n}$ as defined in the previous section. The proofs
use a special case of the well-known Pumping Lemma for regular languages \cite{Rabin:Scott:ibmjrd:1959}. Note that this is
used in the context of propositional dynamic logics over context-free languages in the setting where the alphabet is only
unary so that context-free languages boil down to regular ones anyway.

\subsection{The Pumping Lemma for Unary Languages}

The Pumping Lemma for regular languages states that, for any regular language $L$, there is $n$ 
such that any word $w \in L$ with $|w| \geq n$ can be partitioned into $w = uvx$ with $|uv| \leq n$
 and $|v| \geq 1$ such that $uv^i x \in L$ for all $i \in \Nat$. This follows from the fact that 
there must be a finite automaton for $L$, a DFA $\mathcal{A}$ in fact, that accepts it. Any run that
 over a word that is longer that the number of states of $\mathcal{A}$, some state must be visited 
more than once, and the section of the word between these occurrences can be ``pumped'' (up or down).  However,
 the partition depends on the word in question. Moreover, given several regular languages, the 
partitions can differ even for words that are in the intersection of all the languages. The situation
 becomes more predictable in the setting of unary alphabets, i.e.\ those of the form $\{a\}$, in which
 case pumping constants can be found that work for all languages simultaneously. Before we show this, 
we need the following definition:
\begin{definition}
\label{def:transprof}
Let $L_1,\dotsc,L_n$ be regular languages over the same alphabet $\Sigma$, and for $1 \leq i \leq n$, 
let $\aut{A}_i = (Q_i, \Sigma, \delta_i, q^i_I, F_i)$ be finite automata for $L_i$, respectively. A \emph{simultaneous
transition profile} for $\aut{A}_1,\dotsc,\aut{A}_n$ is a relation 
$\tau \subseteq \bigcup_{1\leq i \leq n} (Q_i \times Q_i)$. Each word $w \in \Sigma^*$ defines such a simultaneous transition 
profile via $\tau_a = \bigcup_{1 \leq i \leq n} \{(q,q') \mid (q, a, q') \in \delta_i \}$, 
and $\tau_{av} = \tau_a \tau_v = \bigcup_{1 \leq i \leq n} \{ (q, q') \mid \text{ ex. } q'' \text{ s.t. } (q,q'') \in \tau_a, (q'',q') \in \tau_v\}$.
\end{definition}
Note that  if $q,q'$ are states in $Q_i$, then $(q, q') \in \tau_w$ iff there is a run of $\aut{A}_i$ from $q$ to $q'$ over $w$. 

\begin{lemma}
\label{lem:simpump}
Let $L_1,\dotsc,L_n$ be regular $\Sigma$-languages, and let $a \in \Sigma$. Then there are $m$ and 
$k>0$ such that, for any $l \geq m+k$ and for all $1 \leq i \leq n, j \in \Nat$ we have that 
$a^{l} \in L_i$ iff $a^{l+j\cdot k} \in L_i$.
\end{lemma}
\begin{proof}
Let $\aut{A}_1,\dotsc,\aut{A}_n$ be finite automata for $L_1,\dotsc,L_n$ as in Def.~\ref{def:transprof}. 
By cardinality reasons, no more than $2^{|Q_1|^2 + \dotsb + |Q_n|^2}$ many transition profiles exist. 
Hence, if one enumerates the transition profiles for $a, a^2, a^3$ etc., upon reaching a word of the 
length $2^{|Q_1|^2 + \dotsb + |Q_n|^2} + 1$, one transition profile $\tau$ must have occurred twice. 
Note that, if a transition profile occurs twice, the entire sequence of profiles between the two
occurrences will occur again, i.e.\ the sequence of transition profiles is ultimately periodic.
Let $m$ be such that $\tau = \tau_m$ is the first occurrence of this profile, and $k$ such that 
$\tau_{m+k}$ is the second occurrence of this profile. Then $m$ and $k$ are as in the lemma: let 
$l \geq m+k$ and let $w^l \in L_i$. Then $\aut{A}_i$ has an accepting run over $w^l$, i.e.\ a run from
 $q_I^i$ to $q \in F_i$. Let $q'$ be the state in the run after reading $a^{m}$. By definition of $q'$, 
there is a run of $\aut{A}_i$ over $a^{l-m}$ from $q'$ to $q$. Since $\tau_{m} = \tau_{m+k}$, there 
is also a run of $\aut{A}_i$ from $q_I^i$ to $q'$ over $a^{m+k}$. By combining these two runs, we obtain
 a run from $q_I^i$ to $q'$ over $a^{m+k}$, and then a run from $q'$ to $q$ over $a^{l-m}$, which is an
 accepting run over $a^{m+k+l-m} = a^{l+k}$. The rest of the claim is by repeated application of the 
previous argument. 
\end{proof}

\begin{figure}
\begin{tikzpicture}[->, node distance = 1.7cm]
                \node[circle, draw] (14)  {$l$};
                \node[] (15) [right of=14] {\dots};
                \node[circle, draw, label=45:{$p$}] (16) [right of=15] {$0$};
    			\node[circle, draw] (17) [right  = 2cm  of 16] {$l+k$};
    			\node[] (18) [right of=17] {\dots};
                \node[circle, draw] (19) [right of=18] {$l$};
                \node[] (21) [right of=19] {\dots};
                \node[circle, draw, label=45:{$p$}] (22) [right of=21] {$0$};
                \node (23) [left of=14] {};
								\node (24) [left of=17] {};
    
                \path
                (14) edge [above, bend left=20] node {$b$} (15)
                (15) edge [above, bend left=20] node {$b$} (16)
                (17) edge [above, bend left=20] node {$b$} (18)
                (18) edge [above, bend left=20] node {$b$} (19)
                (19) edge [above, bend left=20] node {$b$} (21)
                (21) edge [above, bend left=20] node {$b$} (22)
				(14.south east)
        edge[-, decorate,decoration={brace,mirror,raise=.15cm},"$l \geq (m+k)\cdot d$"below=6pt]
        (14.south east -| 16.south west)
        (17.south east)
        edge[-, decorate,decoration={brace,mirror,raise=.15cm},"$k$"below=6pt]
        (17.south east -| 19.south west)
        (19.south east)
        edge[-, decorate,decoration={brace,mirror,raise=.15cm},"$l \geq (m+k)\cdot d$"below=6pt]
        (19.south east -| 22.south west)
        
        (23) edge [above] node {$\Transsys^b_l$} (14)
        (24) edge [above] node {$\Transsys^b_{l+k}$} (17)
                ;
            \end{tikzpicture}
						\caption{Transition Systems $\Transsys^b_{l}$ and $\Transsys^b_{l+k}$ for $l \geq (m+k)\cdot d$.}
						\label{fig:b-long}
						\end{figure}
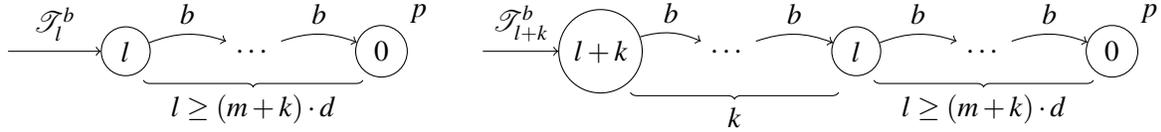

For the following lemma, let $\Progs = \{a,b\}, \Prop = \{p\}$ and, for $l \in \Nat$, let 
$\Transsys^b_l$ be defined as $\Transsys^b_l = (\{0,\dotsc,l\},$ $\Transition{}{}{}, \Label)$ 
with $\Transition{}{}{} = \{(i+1, b, i) \mid 0 \leq i \leq l-1\}$ and $\Label(s) = \{p\}$ iff $s = 0$.
See Fig.~\ref{fig:b-long} for a graphical representation.

\begin{lemma}
\label{lem:pumping-b-reg}
Let $\lclass = \{L_1,\dotsc,L_n\}$ where $L_1,\dotsc,L_n$ are regular $\Progs$-languages, $m$ and $k>0$ be their combined
pumping indices according to Lemma~\ref{lem:simpump}, $d \in \Nat$, $0 < d' \leq d$ and $l \geq (m+k)\cdot d$. 
Then, for all $l \geq j \geq (m+k)\cdot d'$, the states $j$ in $\Transsys^b_{l}$ 
and $j+k$ in $\Transsys^b_{l + k}$ satisfy the same $\PDL{\lclass}$ formulas of modal depth 
at most $d'$.
\end{lemma}
\begin{proof}
The proof is by induction over $d'$.  Assume that the result has been proven for all $0 < d'' \leq d'$.
Let $j \geq (m+k)\cdot d'$. 
Let $\varphi = \mudiam{L_i}\psi$ with $\md{\psi} \leq d'-1$ and $1 \leq i \leq n$. Assume that state 
$j$ in $\Transsys_{l}$ satisfies $\varphi$. Then there is $w \in \{b\}^*$ with $|w| \leq j$ 
such that $w \in L_i$ and the state $j-|w|$ satisfies $\psi$. There are two cases: If $|w| < m+k$, 
then $j-|w| \geq (m+k) \cdot (d'-1)$. If $d' = 1$, the result is immediate, since $\Label(s) = \emptyset$
unless $s=0$ in either LTS. If $d' > 1$, we can use the induction hypothesis to infer that also state 
$j-|w| + k$ in $\Transsys^b_{l + k}$ satisfies $\psi$. If $|w| \geq m+k$, then, by 
Lemma~\ref{lem:simpump}, also $b^{|w|+k} \in L_i$, whence state $j +k$ also satisfies $\mudiam{L_i}\psi$ 
in $\Transsys^b_{l+ k}$. 

Conversely, let $\varphi = \mudiam{L_i}\psi$ hold at state $j+k$ in $\Transsys^b_{l + k}$. 
Then there is $w \in \{b\}^*$ with $|w| \leq j+k$ such that $w \in L_i$ and the state $j+k-|w|$ satisfies 
$\psi$. Again, there are two cases. If $|w| < m+k$, then, $j+k-|w| \geq (m+k) \cdot (d'-1) + k$. If $d' = 1$
we again refer to the fact that $\Label(s) = \emptyset$ unless $s = 0$. If $d' > 1$  we can 
use the induction hypothesis to infer that also state $j-|w|$ satisfies $\psi$ in $\Transsys^b_{l}$.  
If $|w| \geq m+k$, then, by Lem.~\ref{lem:simpump}, also $b^{|w|-k} \in L_i$, whence state $j$ also satisfies 
$\mudiam{L_i}\psi$ in $\Transsys^b_{l}$.
\end{proof}

\begin{figure}
\begin{tikzpicture}[->, node distance = 1.7cm]
                \node[circle, draw] (1) [] {\scriptsize{$l_4$}};
                \node[] (2) [right of=1] {\dots};
                \node[circle, draw] (3) [right of=2] {\scriptsize{$0_4$}};
                \node[circle, draw] (4) [right of=3] {\scriptsize{$u$}};
                \node[circle, draw] (5) [right = 1cm of 4] {\scriptsize{$l_1$}};
                \node[] (6) [right  of=5] {\dots};
                \node[circle, draw] (7) [right of=6, label=45:{$p$}] {\scriptsize{$0_1$}};
                
                \node[circle, draw] (11) [below = 3cm of 1] {\scriptsize{$l_5$}};
                \node[] (20) [right of=11] {\dots};
                \node[circle, draw] (12) [right of=20] {\scriptsize{$0_5$}};
                \node[circle, draw] (13) [right of=12] {\scriptsize{$d$}};
                \node[circle, draw] (14) [above right = 0.5cm and 1cm of 13] {\scriptsize{$l_2$}};
                \node[] (15) [right of=14] {\dots};
                \node[circle, draw] (16) [right of=15, label=45:{$p$}] {\scriptsize{$0_2$}};
    			\node[circle, draw] (17) [below right = 0.5cm and 1cm of 13] {\scriptsize{$l'_3$}};
    			\node[] (18) [right of=17] {\dots};
                \node[circle, draw] (19) [right of=18] {\scriptsize{$l_3$}};
                \node[] (21) [right of=19] {\dots};
                \node[circle, draw] (22) [right of=21, label=45:{$p$}] {\scriptsize{$0_3$}};
                \node[] (23) [left = 0.8cm of 1] {};
                \node[] (24) [left = 0.8cm of 11] {};

                \path
                (1) edge [above, bend left=20] node {$a$} (2)
                (2) edge [above, bend left=20] node {$a$} (3)
                (3) edge [above, bend left=20] node {$a$} (4)
                (4) edge [above, bend left=20] node {$b$} (5)
                (5) edge [above, bend left=20] node {$b$} (6)
                (6) edge [above, bend left=20] node {$b$} (7)
                (1.south east)
                edge[-, decorate,decoration={brace,mirror,raise=.15cm},"$l = (m+k)\cdot d$"below=6pt]
        (1.south east -| 3.south west)
				(5.south east)
        edge[-, decorate,decoration={brace,mirror,raise=.15cm},"$l = (m+k)\cdot d$"below=6pt]
        (5.south east -| 7.south west)
        
        (11) edge [above, bend left=20] node {$a$} (20)
                (20) edge [above, bend left=20] node {$a$} (12)
                (12) edge [above, bend left=20] node {$a$} (13)
                (13) edge [above, bend left=20] node {$b$} (14)
                (13) edge [below, bend right=20] node {$b$} (17)
                (14) edge [above, bend left=20] node {$b$} (15)
                (15) edge [above, bend left=20] node {$b$} (16)
                (17) edge [above, bend left=20] node {$b$} (18)
                (18) edge [above, bend left=20] node {$b$} (19)
                (3) edge [right] node {$a$} (13)
                (19) edge [above, bend left=20] node {$b$} (21)
                (21) edge [above, bend left=20] node {$b$} (22)
                (11.south east)
                edge[-, decorate,decoration={brace,mirror,raise=.15cm},"$l = (m+k)\cdot d$"below=6pt]
        (11.south east -| 12.south west)
				(14.south east)
        edge[-, decorate,decoration={brace,mirror,raise=.15cm},"$l = (m+k)\cdot d$"below=6pt]
        (14.south east -| 16.south west)
        (17.south east)
        edge[-, decorate,decoration={brace,mirror,raise=.15cm},"$k$"below=6pt]
        (17.south east -| 19.south west)
        (19.south east)
        edge[-, decorate,decoration={brace,mirror,raise=.15cm},"$l = (m+k)\cdot d$"below=6pt]
        (19.south east -| 22.south west)
        
        (23) edge [above] node {$\Transsys^{m,k,d}_{1}$} (1)
        (24) edge [above] node {$\Transsys^{m,k,d}_2$} (11)
                ;
            \end{tikzpicture}
						\caption{Transition Systems $\Transsys^{m,k,d}_{1}$ and $\Transsys^{m,k,d}_2$.}
						\label{fig:diaboxstr}
						\end{figure}
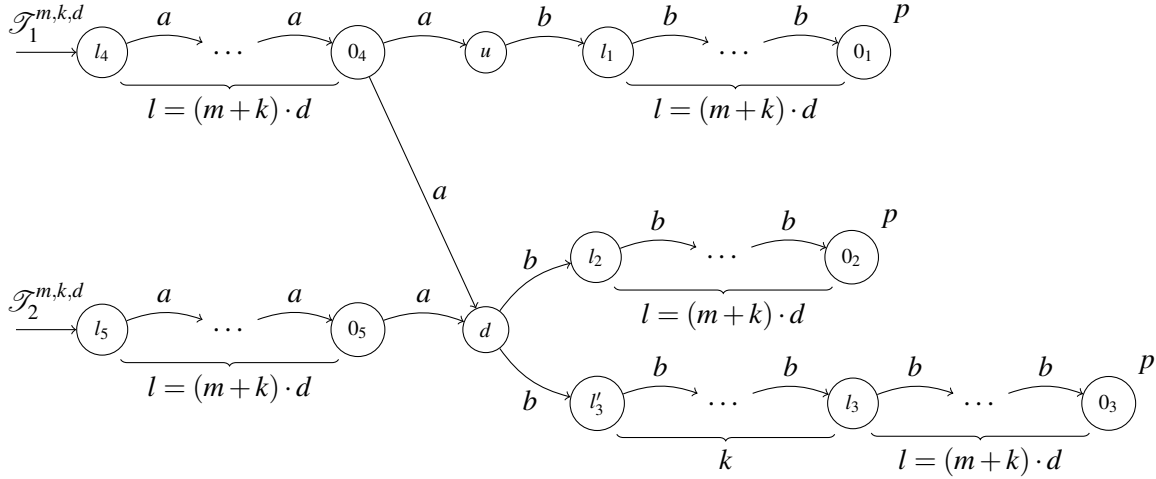

\begin{definition}
\label{def:str-ab}
Let $m,k,d > 1$, $l = (m+k)\cdot d$ and $l' = l+k$. The LTS $\Transsys^{m,k,d}_{1}$ and $\Transsys^{m,k,d}_2$ are defined as in Fig.~\ref{fig:diaboxstr}.
\end{definition}

Let $\Progs = \{a, b\}$. Our aim is to show that the property $\mudiam{a^n}\mubox{b^n}$ cannot be expressed in $\PDLCFL$. 

\begin{lemma}
\label{lem:diabox-indist}
Let $\lclass = \{L_1,\dotsc,L_n\}$ be a collection of context-free languages and, for $1 \leq i \leq n$, 
let $L'_i = L_i \cap L(b^*)$ be the (regular) intersection of $L_i$ with $\{b\}^*$.  Let $m$ and $k$ be their 
combined pumping indices as per Lemma~\ref{lem:simpump} applied to $L'_1,\dotsc,L'_n$. Let $d \geq 1$,  
$l = (m+k)\cdot d$ and $l' = l+k$. 

Suppose $\Transsys_1$ and $\Transsys_2$ are the LTS $\Transsys^{m,k,d}_1$, respectively $\Transsys^{m,k,d}_2$ 
as per Def.~\ref{def:str-ab}. Then no $\PDL{\lclass}$ formula of modal depth $d$ or less 
distinguishes $\Transsys_1$ and $\Transsys_2$.
\end{lemma}
\begin{proof}
Clearly, for $0 \leq j \leq l$, the states $j_1, j_2$ and $j_3$ on the right halves of $\Transsys_1$ and $\Transsys_2$
 satisfy the same formula of any modal depth due 
to the natural isomorphism of the associated structures. Moreover, by Lemma~\ref{lem:pumping-b-reg}, the states 
$j_1, j_2$ and $(j+k)_3$ satisfy the same $\PDL{L'_1,\dotsc,L'_n}$-formulas of modal depth $d' > 0$ if $j \geq (m+k) \cdot d'$,
and, since the associated sub-LTS contain only $b$-transitions, also the same $\PDL{\lclass}$ formulas
of modal depth at most $d'$. By another application of Lemma~\ref{lem:pumping-b-reg}, we obtain that the states $u$ and $d$ 
also satisfy the same $\PDL{\lclass}$ formulas of modal depth at most $d$ since they satisfy the same formulas of the form
$\mudiam{L_i} \psi$.

Towards the claim of the lemma, it remains to show that, for all $0 \leq j \leq l$, the states $l_4$ and $l_5$ satisfy the 
same formulas of the form $\mudiam{L_i} \psi$, where $1 \leq i \leq n$ and $\md{\psi} \leq d$. Let $0 \leq j \leq l$ and 
assume that this has been shown for all $j' < j$. Let $1 \leq i \leq n$ and $\md{\psi} \leq n$.
Consider $\mudiam{L_i}\psi$. In order for it to hold at state $j_4$, respectively $j_5$, there must be $w \in L_i$ and
some other state reachable via $w$ such that $\psi$ holds at that state. In case that $w \in \{a\}^*$, this state is either $u$,
respectively $d$, for which the result follows immediately, or the state is $j'_4$, respectively $j'_5$ with $j' \geq j$. 
In this case, the result follows from the induction hypothesis.
Hence, the remaining case is that where $w \in \{a^+ b^+\}$ and the witness for $\psi$ is one of the $j'_1, j'_2, j'_3$. 
For the latter two cases, note that any $j'_2$ or $j'_3$ is reachable from $j_4$ and $j_5$ via the exact same word, whence
the claim immediately follows. 

Hence, the interesting case is that where $\mudiam{L_i}\psi$ holds at $j_4$ due to $\psi$ 
holding at $j'_1$ with $w \in L_i$ labelling the path from $j_4$ to $j'_1$. However, note that $w$ also labels the path 
from $j_5$ to $j'_2$, at which $\psi$ also holds due to the sub-structures reachable from $j'_1$ and $j'_2$ being isomorphic. 
Hence, $j_4$ and $j_5$ satisfy exactly the same $\PDL{\lclass}$ formulas of the form $\mudiam{L_i}\psi$ modal depth at most $d$,
and, hence exactly the same $\PDL{\lclass}$ formulas of modal depth $d$. In particular, this holds for $l_4$ and $l_5$, which
is the claim of the lemma.
\end{proof}

\begin{theorem}
\label{thm:diabox}
Property $\mudiam{a^n}\mubox{b^n}$ cannot be expressed in $\PDLCFL$. 
Hence, $\vpFLC \not\le \PDLCFL$.
\end{theorem}
\begin{proof}
Assume to the contrary that there is $\varphi \in \PDLCFL$ that expresses $\mudiam{a^n}\mubox{b^n}$, i.e.\ $\varphi$ holds 
in exactly those LTS that satisfy it. Let $d$ be the modal depth of $\varphi$ and let $\lclass = L_1,\dotsc,L_n$
be a list of the languages used in $\varphi$. By Lemma~\ref{lem:diabox-indist}, $\varphi$ cannot distinguish
the transition systems $\Transsys_1$ and $\Transsys_2$ in Fig.~\ref{fig:diaboxstr}. However, clearly
$\Transsys_1$ satisfies $\mudiam{a^n}\mubox{b^n}$, while $\Transsys_2$ does not. Hence, $\mudiam{a^n}\mubox{b^n}$ cannot be expressed
in $\PDLCFL$.

Conversely, $\mudiam{a^n}\mubox{b^n}$ can be expressed in $\vpFLC$ as seen in Sec.~\ref{sec:sepprop}.
\end{proof}
Note that, in fact, the result can be strengthened to the Boolean closure of context-free languages, since the intersection
of a language in the boolean closure of $\CFL$ with a unary alphabet remains regular. Moreover, the theorem also separates
$\PDLVPL$ from \vpFLC since the former is a fragment of \PDLCFL.

The previous theorem already supports the intuition that the reason for the 
ineffability of $\mudiam{a^n}\mubox{b^n}$ is not to be found in language-theoretic reasons, but
in the alternation of the modal operators, which in some sense ``insulates'' the front
part and the back part of the property, i.e.\ $\mudiam{a^n}$, respectively $\mubox{b^n}$
from each other, preventing the constraint on the joint number of letters $n$ to 
be ``remembered'' in the rest of the formula. In order to underline this point we now sketch 
that the property $\mudiam{a^n}\mubox{b}\mudiam{a^n}$ 
can also not be expressed in $\PDLCFL$. 

\begin{theorem}
\label{thm:an-b-an}
Property $\mudiam{a^n}\mubox{b}\mudiam{a^n}$ cannot be expressed in $\PDLCFL$. Hence, $\FLC \not \le \PDLCFL$.
\end{theorem}
\begin{proof} (sketch)
Consider the structures $\Transsys^a_1$ and $\Transsys^a_2$ in Fig.~\ref{fig:diaboxstr2} for 
$m, k, d$ to be given later. Clearly
$\Transsys^a_1$ satisfies $\mudiam{a^n}\mubox{b}\mudiam{a^n}$, while $\Transsys^a_2$ does not. 

The proof that the two structures cannot be distinguished in $\PDLCFL$ proceeds via the 
same pattern as the proof of Thm.~\ref{thm:diabox}, i.e.\ the structures are built depending 
on the context-free languages $L_1,\dotsc,L_n$ used in the hypothetical formula that expressed $\mudiam{a^n}\mubox{b}\mudiam{a^n}$, 
and its modal depth $d$. 
In particular, in the proof that the states $j_1, j_2$ and $(j+k)_3$ satisfy the same formulas 
of a given modal depth over  $L_1,\dotsc,L_n$ proceeds in the same pattern by invoking Lemma~\ref{lem:pumping-b-reg}
over the (regular) intersections of $L_1,\dotsc,L_n$ with $\{a\}^*$. Also, after establishing 
that this holds, and that the states $u$ and $d$ satisfy the same $\PDL{L_1,\dotsc,L_n}$ formulas
 of modal depth $d$, the proof for the left part of the structures 
is the same. 

However, it is not as straightforward to establish that $u$ and $d$ satisfy the same 
$\PDL{L_1,\dotsc,L_n}$-formulas of modal depth $d$, since the paths leading out of $u$ and $d$ are not
over a unary alphabet. On the other hand, all these paths are labelled by a word starting with exactly
one $b$, followed by a number of $a$'s. Hence, we can equivalently replace any $\mudiam{L_i}\psi$ by 
$\mudiam{b}\mudiam{\deriv{b}{L_i}}\psi$ where $\deriv{b}{L_i}$ is the $b$-derivative of $L_i$. Note that, 
by Lemma~\ref{lem:derivatives}, $\deriv{b}{L_i}$ is context-free again and, hence, can now be replaced
by its (regular) intersection with $\{a\}^*$, since it is interpreted over paths that contain only $a$'s. Of course, 
this has to be taken into account when defining $m$ and $k$ via Lemma~\ref{lem:simpump}. It is not hard to see 
that the proof succeeds by defining $m$ and $k$ over the set of languages $L'_1,\dotsc,L'_n, {L'}^b_1,\dotsc,{L'}^b_n$
where $L'_i$ is the intersection of $L_i$ with $\{a\}^*$, and ${L'}^b_i$ is the intersection of $\deriv{b}{L_i}$ with $\{a\}^*$. 

On the other hand, we have seen in Sec.~\ref{sec:sepprop} that $\mudiam{a^n}\mubox{b}\mudiam{a^n}$ can be expressed in \FLC
by the formula $(\mu Z.\mudiam{a};\mubox{b};\mudiam{a} \vee \mudiam{a};Z;\mudiam{a});p$.
\end{proof}

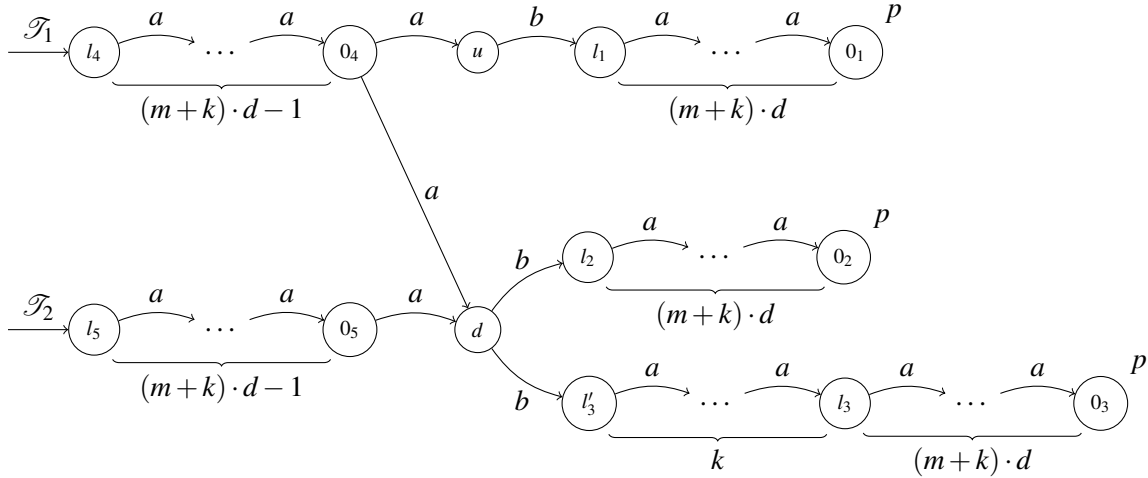
\begin{figure}
\begin{tikzpicture}[->, node distance = 1.7cm]
                \node[circle, draw] (1) [] {\scriptsize{$l_4$}};
                \node[] (2) [right of=1] {\dots};
                \node[circle, draw] (3) [right of=2] {\scriptsize{$0_4$}};
                \node[circle, draw] (4) [right of=3] {\scriptsize{$u$}};
                \node[circle, draw] (5) [right = 1cm of 4] {\scriptsize{$l_1$}};
                \node[] (6) [right  of=5] {\dots};
                \node[circle, draw] (7) [right of=6, label=45:{$p$}] {\scriptsize{$0_1$}};
                
                \node[circle, draw] (11) [below = 3cm of 1] {\scriptsize{$l_5$}};
                \node[] (20) [right of=11] {\dots};
                \node[circle, draw] (12) [right of=20] {\scriptsize{$0_5$}};
                \node[circle, draw] (13) [right of=12] {\scriptsize{$d$}};
                \node[circle, draw] (14) [above right = 0.5cm and 1cm of 13] {\scriptsize{$l_2$}};
                \node[] (15) [right of=14] {\dots};
                \node[circle, draw] (16) [right of=15, label=45:{$p$}] {\scriptsize{$0_2$}};
    			\node[circle, draw] (17) [below right = 0.5cm and 1cm of 13] {\scriptsize{$l'_3$}};
    			\node[] (18) [right of=17] {\dots};
                \node[circle, draw] (19) [right of=18] {\scriptsize{$l_3$}};
                \node[] (21) [right of=19] {\dots};
                \node[circle, draw] (22) [right of=21, label=45:{$p$}] {\scriptsize{$0_3$}};
                \node[] (23) [left = 0.8cm of 1] {};
                \node[] (24) [left = 0.8cm of 11] {};

                \path
                (1) edge [above, bend left=20] node {$a$} (2)
                (2) edge [above, bend left=20] node {$a$} (3)
                (3) edge [above, bend left=20] node {$a$} (4)
                (4) edge [above, bend left=20] node {$b$} (5)
                (5) edge [above, bend left=20] node {$a$} (6)
                (6) edge [above, bend left=20] node {$a$} (7)

                (1.south east)
                edge[-, decorate,decoration={brace,mirror,raise=.15cm},"$(m+k)\cdot d-1$"below=6pt]
        (1.south east -| 3.south west)
				(5.south east)
        edge[-, decorate,decoration={brace,mirror,raise=.15cm},"$(m+k)\cdot d$"below=6pt]
        (5.south east -| 7.south west)
        
        (11) edge [above, bend left=20] node {$a$} (20)
                (20) edge [above, bend left=20] node {$a$} (12)
                (12) edge [above, bend left=20] node {$a$} (13)
                (13) edge [above, bend left=20] node {$b$} (14)
                (13) edge [below, bend right=20] node {$b$} (17)
                (14) edge [above, bend left=20] node {$a$} (15)
                (15) edge [above, bend left=20] node {$a$} (16)
                (17) edge [above, bend left=20] node {$a$} (18)
                (18) edge [above, bend left=20] node {$a$} (19)
                (3) edge [right] node {$a$} (13)
                (19) edge [above, bend left=20] node {$a$} (21)
                (21) edge [above, bend left=20] node {$a$} (22)
                (11.south east)
                edge[-, decorate,decoration={brace,mirror,raise=.15cm},"$(m+k)\cdot d-1 $"below=6pt]
        (11.south east -| 12.south west)
				(14.south east)
        edge[-, decorate,decoration={brace,mirror,raise=.15cm},"$(m+k)\cdot d$"below=6pt]
        (14.south east -| 16.south west)
        (17.south east)
        edge[-, decorate,decoration={brace,mirror,raise=.15cm},"$k$"below=6pt]
        (17.south east -| 19.south west)
        (19.south east)
        edge[-, decorate,decoration={brace,mirror,raise=.15cm},"$(m+k)\cdot d$"below=6pt]
        (19.south east -| 22.south west)
        
        (23) edge [above] node {$\Transsys_1$} (1)
        (24) edge [above] node {$\Transsys_2$} (11)
                ;
            \end{tikzpicture}
						\caption{Transition Systems $\Transsys^a_1$ and $\Transsys^a_2$.}
						\label{fig:diaboxstr2}
						\end{figure}

\section{Separating Decidable and Undecidable Logics}
\label{sec:sepdecfromundec}
% !TEX root =  main.tex

We now show that the inclusions of decidable program logics in Fig.~\ref{fig:hierarchy} in the undecidable ones are strict.
More precisely, we show that \PDLVPL is strictly less expressive than \PDLCFL, and that \vpFLC is strictly less expressive 
than \FLC, which separates the top band in Fig.~\ref{fig:hierarchy} from the middle band. 

It is tempting to conclude an expressivity gap from the decidability gap; however, only a
weaker proposition follows: there can be no \emph{computable} equivalence-preserving translation from \FLC to \vpFLC, resp.\ from
\PDLCFL to \PDLVPL. This does not preclude the existence of equivalent formulas in the smaller logic for each of the larger, though. It
merely says that they could not be constructed effectively if they exist.

However, with a little bit more observation it is possible to extend this to an expressiveness gap, too. 

\begin{lemma}
\label{lem:anbaninvpflc}
Let $p$ be an atomic proposition. Suppose $\mudiam{a^nba^n}p$ was expressible in \vpFLC. Then its satisfiability problem would
be undecidable. 
\end{lemma}

\begin{proof}
Harel et al.\ present a reduction from Post's Correspondence Problem (PCP) to the satisfiability problem for $\PDL{a^nba^n}$
\cite{JCSS::HarelPS1983}. They
show how to construct, for every input $\mathcal{I} = \{(u_1,v_1),\ldots,(u_n,v_n)\}$ to PCP, a formula $\varphi_{\mathcal{I}}$ of
$\PDL{a^nba^n}$ which is satisfiable iff $\mathcal{I}$ has a solution (in the sense of PCP). 

Now suppose there was a \vpFLC-formula $\varphi_{a^nba^n}$ equivalent to $\mudiam{a^nba^n}p$. Since any ordinay \PDL-formula can
easily be expressed in \vpFLC, we immediately get a similar reduction from PCP to \vpFLC's satisfiability problem: for any instance 
$\mathcal{I}$ of $\varphi$ we can construct a \vpFLC-formula $\varphi'_{\mathcal{I}}$ in the same way as the 
$\PDL{a^nba^n}$-formula $\varphi_{\mathcal{I}}$ with the only difference that we use $\varphi_{a^nba^n}[\psi/p]$ whenever 
$\varphi_{\mathcal{I}}$ uses $\mudiam{a^nba^n}\psi$. Clearly, $\varphi'_{\mathcal{I}}$ is equivalent to $\varphi_{\mathcal{I}}$ 
for any $\mathcal{I}$, and so it is satisfiable iff $\mathcal{I}$ is solvable. 
\end{proof}

Note that this argument relies only on the existence of a formula equivalent to $\mudiam{a^nba^n}p$, not its effective constructibility.
In fact, the question after the effective constructibility of $\varphi_{a^nba^n}$ is meaningless as it is a \emph{fixed} formula
and is therefore trivially constructible whenever it exists, as for every fixed formula there is clearly an algorithm which can
write it down. Instead, it is the modularity of \vpFLC, i.e.\ the possibility to build formulas by replacing subformulas, which is 
used in order to handle arbitrary PCP inputs $\mathcal{I}$.

\begin{theorem}
\label{thm:sepundec}
We have \FLC $\not\le$ \vpFLC and \PDLCFL $\not\le$ \PDLVPL.
\end{theorem}

\begin{proof}
Clearly, $\mudiam{a^nba^n}p$ can be expressed in \PDLCFL. By \cite{langesomla-ipl06}, it is expressible in \FLC. On the other hand, if 
it was expressible in \vpFLC then, by Lemma~\ref{lem:anbaninvpflc}, \vpFLC's satisfiability problem would be undecidable contradicting
its decidability result from \cite{conf/concur/BruseL21}.

Likewise $\mudiam{a^nba^n}p$ cannot be expressible in \PDLVPL either, as it would then be expressible in \vpFLC, too, by the
generic embedding of \PDLCFL into \FLC \cite{langesomla-ipl06} which produces formulas from \vpFLC when applied to formulas 
from \PDLVPL. Equally, the contradiction can be obtained using the decidability result for \PDLVPL \cite{journals/jlp/LodingLS07}.
\end{proof}

So the different status of decidability between program logics does not immediately yield a gap in expressiveness, but it can be used
to construct one by embedding presumably inexpressible properties in a set of formulas such that its subset of satisfiable ones 
is decidable in one case and undecidable in the other.

\section{Conclusion}
\label{sec:concl}
% !TEX root =  main.tex

\paragraph*{Summary.}
We have completely mapped the structures in the hierarchy of expressiveness amongst program logics for non-regular properties
up to context-free ones. The two main strands of logics for such purposes found in the literature are propositional
dynamic ones which incorporate formal languages into modal operators, and modal fixpoint logics which can, to some
extent, mimic the generation of formal languages through least and greatest fixpoint constructions. We have 
provided formal proofs of what one may expect, namely that the bounded modality alternation inherent in propositional
dynamic logics cannot be overcome by subtle constructions: there are properties which require some -- even the minimal
-- amount of alternation amongst modal operators which cannot be expressed in these propositional dynamic logics.
Note that $\mudiam{a^n}\mubox{b^n}$ only features one swap from an existential to a universal modality, and
$\mudiam{a^n}\mubox{b}\mudiam{a^n}$ features the smallest possible amount of one kind of operator: only a single box-modality.

\paragraph*{Further Work.}
One can, surely, devote an arbitrary amount of time to find further properties that witness the separation of logics
presented here. For instance, with the developments leading to Thm.~\ref{thm:diabox} it should not be too difficult to
show that $\mubox{a^n}\mudiam{b^n}$ cannot be expressed in \PDLCFL either. 

Far more interesting, though, would be to investigate whether such separation techniques could be applied even further up the
hierarchy of program logics. Note that \PDL{} is a very generic formalism that is formally defined for \emph{any} language
class. Thus, any hierarchy of language classes imposes a hierarchy of \PDL{}-logics, but strictness amongst languages does
not immediately transfer to the logics. Instead, more or less sophisticated arguments are needed. As shown here, the argument
based on the Pumping Lemma can be used up to the context-free languages, in fact even their Boolean closure. Beyond, for 
instance for the class \CSL\ of context-sensitive languages, it is unclear whether there are separation results to be discovered
in a similar style. So a separation of \PDL{\CSL} from \PDLCFL for instance has, as far as we know, not been shown yet.

One may argue that beyond \PDLCFL and \FLC, the question of the strictness of the hierarchy becomes less interesting as
these logics are undecidable already. There is, however, still a vast space of program logics with potential applications in
formal verification despite undecidability of their satisfiability problems, as decidability of their model checking problems
reaches far beyond that. On the modal fixpoint logic strand, even full \HFL{} -- which lifts \mucalc not only to predicate
transformers as \FLC does, but also to higher-order predicate transformers of arbitrary arity -- retains model checking 
decidability, albeit of complexity that is $k$-fold exponential in the size of the underlying LTS \cite{als-mchfl07} when 
$k$ equals the maximal type order of such transformers. 

The complexity of model checking propositional dynamic logics is well grounded in formal language theory, as it is polynomially
linked to the complexity of the emptiness problem for intersections with regular languages \cite{la-reachpdl:2011}, yielding,
for instance exponential-time model checking for \PDL{} over indexed languages \cite{Aho68}, and doubly exponential-time model
checking for \PDL{} over multi-stack visibly pushdown languages \cite{conf/lics/TorreMP07}.

A natural question that arises from the lifting of the decidability gap in satisfiability checking to the expressiveness gap, 
as done in the previous section, is whether complexity-theoretic gaps can be used for such purposes as well. The answer is
of course yes: if the data complexity of two logics is separated by provably different complexity classes then so is their
expressivity. This has been used for instance to establish that each \HFL{k+1} is more expressive than \HFL{k} for $k \ge 1$
\cite{als-mchfl07}, making use of the time hierarchy theorem. Likewise, the space hierarchy theorem can be used to
separate the so-called tail-recursive fragments of each \HFL{k} \cite{BruseLL21}. It remains to be seen, though, if such results
can be used to obtain separations from highly expressive \PDL{}-based logics.

\bibliographystyle{eptcs}
\bibliography{biblio}
\end{document}